\documentclass[11pt]{amsart}
\usepackage{a4wide}
\usepackage[english]{babel}
\usepackage{cite}
\usepackage{amsmath, amsfonts, amssymb,graphicx}
\usepackage{mathrsfs,url}
\usepackage{tkz-graph}
\usepackage[foot]{amsaddr}
\usepackage[nodayofweek]{datetime}
\usepackage{enumerate}
\usepackage[hmargin=1.0in,vmargin=1.0in]{geometry}

\newcommand{\myparagraph}[1]{\noindent{\textbf{#1}}}

\newcommand{\feas}[0]{{\sf Feas}}

\newcommand{\tuple}[1]{\ensuremath{\langle #1\rangle}}

\newcommand{\R}{\mbox{$\mathbb R$}}

\newcommand{\Qnn}{\mbox{$\mathbb Q_{\geq 0}$}}
\newcommand{\Qnnc}{\mbox{$\overline{\mathbb Q}_{\geq 0}$}}

\newcommand{\cal}[1]{\mathcal{#1}}
\newcommand{\opt}[0]{{\sf Opt}}

\newcommand{\mm}[2]{\langle #1, #2 \rangle}

\newcommand{\meet}{\wedge}
\newcommand{\join}{\vee}

\newcommand{\supp}{\mbox{\rm supp}}
\newcommand{\fhom}{\rightarrow_f}
\newcommand{\opm}[2]{{\cal O}^{(#2)}_{#1}}

\date{\today}

\author{Johan Thapper}
     \address{Laboratoire d'Informatique  (LIX), CNRS UMR 7161\\
     \'{E}cole Polytechnique \\91128 Palaiseau\\
     France}
     \email{thapper@lix.polytechnique.fr}
 \thanks{The research leading to these results has received funding from the European Research Council under the European Community's Seventh Framework Programme (FP7/2007-2013 Grant Agreement no.\ 257039)}

\author{Stanislav \v{Z}ivn\'{y}}
    \address{Department of Computer Science\\
             University of Oxford, Oxford\\
             UK}
    \email{standa.zivny@cs.ox.ac.uk}
    \thanks{Stanislav \v{Z}ivn\'y was supported by a Junior Research Fellowship at
    University College, Oxford.}

\title[The Power of Linear Programming for Valued CSPs]{The Power of Linear Programming for Valued CSPs}

\newtheorem{theorem}{Theorem}[section]
\newtheorem{lemma}[theorem]{Lemma}
\newtheorem{proposition}[theorem]{Proposition}
\newtheorem{corollary}[theorem]{Corollary}

\newtheorem{remark}{Remark}
\newtheorem{definition}{Definition}
\theoremstyle{remark}
\newtheorem{example}{Example}

\def\multiset#1#2{\big(\kern-.2em\big(\genfrac{}{}{0pt}{}{#1}{#2}\big)\kern-.2em\big)}
\def\mmultiset#1#2{\left(\kern-.3em\left(\genfrac{}{}{0pt}{}{#1}{#2}\right)\kern-.3em\right)}

\begin{document}

\begin{abstract}

A class of \emph{valued constraint satisfaction problems} (VCSPs) is
characterised by a \emph{valued constraint language}, a fixed set of cost
functions on a finite domain. An instance of the problem is specified by a sum
of cost functions from the language with the goal to minimise the sum. This
framework includes and generalises well-studied constraint satisfaction problems
(CSPs) and maximum constraint satisfaction problems (Max-CSPs).

Our main result is a precise algebraic characterisation of valued constraint
languages whose instances can be solved \emph{exactly} by the \emph{basic linear
programming relaxation}.
Using this result, we obtain tractability of several
novel and previously widely-open classes of VCSPs, including problems over
valued constraint languages that are: (1) submodular on \emph{arbitrary
lattices}; (2) bisubmodular (also known as $k$-submodular) on \emph{arbitrary
finite domains}; (3) weakly (and hence strongly) tree-submodular on
\emph{arbitrary trees}. 


\end{abstract}

\maketitle
\newpage

\section{Introduction}

The constraint satisfaction problem (CSP) provides a common framework for many
theoretical and practical problems in computer science. An instance can be
vaguely described as a set of variables to be assigned values from the domains
of the variables so that all constraints are
satisfied~\cite{Montanari74:constraints}. The CSP is NP-complete in general and
thus we are interested in restrictions which give rise to tractable classes of
problems. Following Feder \& Vardi~\cite{Feder98:monotone}, we restrict the
constraint language; that is, all constraint relations in a given instance must
belong to a fixed, finite set of relations on the domain. The most successful
approach to classifying language-restricted CSPs is the so-called algebraic
approach~\cite{Jeavons97:closure,Jeavons98:algebraic,Bulatov05:classifying},
which has led to several complexity
classifications~\cite{Bulatov06:3-elementJACM,Bulatov11:conservative,Barto09:siam,Barto11:lics}
and algorithmic characterisations~\cite{Barto09:focs,Idziak10:siam} going beyond
the seminal work of Schaefer~\cite{Schaefer78:complexity}.

Motivated by reasons both theoretical (optimisation problems are different from
decision problems) and practical (many problems are over-constrained and hence
have no solution, or under-constrained and hence have many solutions), we study
valued constraint satisfaction problems
(VCSPs)~\cite{Bistarelli97:semiring,Schiex95:valued}. A valued constraint
language is a finite set of cost functions on the domain, and a VCSP instance is
given by a weighted sum of cost functions from the language with the goal to
\emph{minimise} the sum. (CSPs correspond to the case when the range of all cost
functions is $\{0,\infty\}$, and Max-CSPs correspond to the case when the range
of all cost functions is $\{0,1\}$.\footnote{With respect to exact solvability,
Max-CSPs (``maximising the number of satisfied constraints'') are
polynomial-time equivalent to Min-CSPs (``minimising the number of unsatisfied
constraints''). Therefore, with respect to exact solvability, Max-CSPs are
polynomial-time equivalent to $\{0,1\}$-valued VCSPs.}) The VCSP framework is
very robust and has also been studied under different names such as Min-Sum
problems, Gibbs energy minimisation, Markov Random Fields (MRF), Conditional
Random Fields (CRF) and others in several different contexts in computer
science~\cite{Lauritzen96,Wainwright08,Crama11:book}. 

Given the generality of the VCSP, it is not surprising that only few complexity
classifications are known. In particular, only Boolean (on a 2-element domain)
languages~\cite{Cohen06:complexitysoft,cz11:cp-mwc} and conservative (containing
all $\{0,1\}$-valued unary cost functions) languages~\cite{kz12:soda} have been
completely classified with respect to exact solvability. On the algorithmic
side, most known tractable languages are somewhat related to submodular
functions on distributive
lattices~\cite{Cohen06:complexitysoft,Cohen08:Generalising,Jonsson11:cp,kz12:soda}.

An alternative approach for solving VCSPs is using linear programming (LP) and
semidefinite programming (SDP); these have been used mostly for
approximation~\cite{Raghavendra08:stoc,Kun12:itcs,Dalmau11:robust,Barto12:stoc}.

\smallskip

\myparagraph{Contribution} 
We study the power of the \emph{basic linear programming relaxation} (BLP). Our
main result (Theorem~\ref{thm:char}) is a precise characterisation of valued constraint languages for
which BLP is a decision procedure. In more detail, we characterise valued
constraint languages over which VCSP instances can be solved exactly by a
certain basic linear program. 
Equivalently, we show precisely when a particular
integer programming formulation of a VCSP has zero integrality gap.
The characterisation is algebraic in terms of \emph{fractional
polymorphisms}~\cite{Cohen06:expressive}.

Our work is the first link between solving VCSPs exactly using LP and the
algebraic machinery for VCSPs introduced by Cohen et al.
in~\cite{Cohen06:expressive,ccjz11:mfcs}. Part of the proof is inspired by the
characterisation of width-1 CSPs~\cite{Feder98:monotone,Dalmau99:set}. One of
the main technical contributions is a construction of totally symmetric
fractional polymorphisms of all arities (Theorem~\ref{thm:gentotot}). 

This result allows us to demonstrate that several valued constraint languages
are covered by our characterisation and thus are tractable; that is,
VCSP instances over these languages can be solved exactly using BLP. New
tractable languages include: (1) submodular languages on \emph{arbitrary
lattices}; (2) bisubmodular (also known as $k$-submodular) languages on
\emph{arbitrary finite domains}; (3) weakly (and hence strongly) tree-submodular
languages on \emph{arbitrary trees}. 
The complexity of (subclasses of) these languages has been mentioned explicitly
as open problems
in~\cite{DeinekoJKK08,Krokhin08:max,Kolmogorov11:mfcs,Huber12:ksub}.
More generally, we show that any valued constraint language with a binary
multimorphism in which at least one operation is a semi-lattice operation is
tractable (cf. Section~\ref{sec:trac}). Our results cover \emph{all} known
tractable finite-valued constraint languages.

\myparagraph{Related work}
Apart from identifying tractable classes of CSPs and VCSPs with respect to exact
solvability, the approximability of Max-CSPs has attracted a lot of
attention~\cite{Creignouetal:siam01,Khanna01:approximability,Jonsson09:tcs}.
Under the assumption of the \emph{unique games conjecture}~\cite{khot10:coco}, 
Raghavendra showed how to
approximate all Max-CSPs and finite-valued VCSPs
optimally~\cite{Raghavendra08:stoc}.\footnote{Note that Max-CSPs
(=$\{0,1\}$-valued VCSPs) and finite-valued VCSPs, respectively, are called CSPs
and Generalised CSPs (GCSPs), respectively, in~\cite{Raghavendra08:stoc}.}  
For VCSPs that are tractable, Raghavendra's algorithms provide a PTAS,
but it seems notoriously difficult to determine the approximation
ratios of these algorithms.
Very
recently, Max-CSPs that are \emph{robustly approximable} have been
characterised as those having \emph{bounded width}~\cite{Dalmau11:robust,Kun12:itcs,Barto12:stoc}.
Specifically, Kun et al.\
studies the question of which Weighted Max-CSPs\footnote{In Weighted
Max-CSPs, every constraint $f$ is $\{0,c_f\}$-valued, where $c_f$ is a positive
constant. Weighted Max-CPSs are a special case of VCSPs.} can be robustly
approximated using BLP~\cite{Kun12:itcs}. 
Their result is related but incomparable to ours as it applies to robust
approximability and not to exact solvability, except for the special case of
width-1 CSPs.
In particular, ``solving'' (``deciding'') for us means finding an optimum
solution to a VCSP instance, which is an optimisation problem, whereas
``solving'' in~\cite{Kun12:itcs} means (ignoring their results on robust
approximability, which do not apply here) the basic LP formulation of a CSP
instance finds a solution if one exists.\footnote{Note that CSPs are defined as
$\{0,1\}$-valued in~\cite{Kun12:itcs} and not as $\{0,\infty\}$-valued, as in
this paper. This is needed for the LP formulation and the measure of
approximability. After all,~\cite{Kun12:itcs} deals with Max-CSPs.}

We remark that our tractability results apply to the minimisation problem of
VCSP instances (i.e., the objective function is given by a sum of ``local'' cost
functions) but not to objective functions given by an oracle. In particular,
submodular functions given by an oracle can be minimised on distributive
lattices~\cite{Schrijver00:submodular,Iwata01:submodular},
diamonds~\cite{Kuivinen11:do-diamonds}, and several constructions on lattices
preserving tractability have been identified~\cite{Krokhin08:max}, but it is
widely open what happens on non-distributive lattices. Similarly, bisubmodular
functions given by an oracle can be minimised in polynomial-time on domains of
size 3~\cite{Fujishige06:bisubmodular}, but the complexity is open on domains of
larger size~\cite{Huber12:ksub}. It is known that strongly tree-submodular
functions given by an oracle can be minimised in polynomial time on binary
trees~\cite{Kolmogorov11:mfcs}, but the complexity is open on general
(non-binary) trees. Similarly, it is known that weakly tree-submodular functions
given by an oracle can be minimised in polynomial time on chains and
forks~\cite{Kolmogorov11:mfcs}, but the complexity on (even binary) trees is
open.

Extending the notion of (generalised) arc consistency for
CSPs~\cite{Mackworth77:consistency,Freuder78:synthesizing} and several
previously studied notions of arc consistencies for
VCSPs~\cite{Cooper04:aij-arc}, Cooper et al. introduced \emph{optimal soft arc
consistency} (OSAC)~\cite{Cooper10:osac}, which is a linear program relaxation
of a given VCSP instance.
Since OSAC is is a tighter relaxation than BLP (cf. Appendix~\ref{sec:osac}),
all tractable classes identified in this paper are solved by OSAC as well.
Similarly, since the basic SDP relaxation from~\cite{Raghavendra08:stoc} is
tighter than BLP, all tractable cases identified in this paper are solved by  it
as well.

\section{Preliminaries}
\label{sec:prelim}

The set of non-negative rational numbers is denoted by $\Qnn$.
A \emph{signature} $\tau$ is a set of \emph{function symbols} $f$, each with an
associated positive \emph{arity}, $ar(f)$. A \emph{valued $\tau$-structure} $A$
(also known as a \emph{valued constraint language}, or just a \emph{language}) 
consists of a \emph{domain} $D = D(A)$, together with a
function $f^A : D^{ar(f)} \rightarrow \Qnn$, for each function symbol $f \in
\tau$. (To be precise, these are \emph{finite-valued} structures. In
Section~\ref{sec:gener-valued}, we will extend $\Qnn$ with infinity.)

Let $A$ be a valued $\tau$-structure. An instance of VCSP$(A)$ is given by a valued $\tau$-structure $I$.
A solution to $I$ is a function $h : D(I) \rightarrow D(A)$,
its measure given by
\[
\sum_{f \in \tau, {\bar x} \in D(I)^{ar(f)}} f^I({\bar x}) f^A(h({\bar x})).
\]
The goal is to find a solution of minimum measure.
This measure will be denoted by $\opt_A(I)$.

For an $m$-tuple $\bar{t}$, we denote by $\{\bar{t}\}$ the set of elements in $\bar{t}$.
Furthermore, we denote by $[\bar{t}]$ the multiset of elements in $\bar{t}$.

\subsection{Fractional Homomorphisms}

Let $A$ and $B$ be valued structures over the same signature $\tau$.
Let $B^A$ denote the set of all functions from $D(A)$ to $D(B)$.
A \emph{fractional homomorphism} from $A$ to $B$ is a function $\omega : B^A \rightarrow \Qnn$, with $\sum_{g \in B^A} \omega(g) = 1$, such that for every function symbol $f \in \tau$ and tuple ${\bar a} \in D(A)^{ar(f)}$, it holds that
\[
\sum_{g \in B^A} \omega(g) f^B(g({\bar a})) \leq f^A({\bar a}),
\]
where the functions $g$ are applied component-wise.

We write $A \fhom B$ to indicate the existence of a fractional homomorphism.

\begin{proposition}\label{prop:frachom}
  Assume that $A \fhom B$.
  Then $\opt_A(I) \geq \opt_B(I)$, for every instance $I$.
\end{proposition}

\begin{proof}
  Let $\omega$ be a fractional homomorphism from $A$ to $B$,
  let $X = D(I)$ and let $h : X \rightarrow A$ be an arbitrary solution. Then,
  \[
  \sum_{f, \bar{x}} f^I({\bar x}) f^A(h({\bar x})) \geq \sum_{f, \bar{x}} f^I({\bar x}) \sum_{g \in B^A} \omega(g) f^B(g(h({\bar x}))) = \sum_{g \in B^A} \omega(g) \sum_{f, \bar{x}} f^I({\bar x}) f^B(g(h({\bar x}))),
  \]
  where the sums are over $f \in \tau$ and $\bar{x} \in X^{ar(f)}$.
  Hence, there exists a $g \in B^A$ such that the measure of the solution $g
  \circ h$ to $I$ as an instance of VCSP$(B)$ is no greater than the measure of
  the solution $h$ to $I$ as an instance of VCSP$(A)$.
\end{proof}

\subsection{Fractional Polymorphisms}

Let $A$ be a valued $\tau$-structure, and let $D = D(A)$. An $m$-ary
\emph{operation} on $D$ is a function $g : D^m \rightarrow D$. Let $\opm{D}{m}$
denote the set of all $m$-ary operations on $D$. An $m$-ary \emph{fractional
operation} is a function $\omega : \opm{D}{m} \rightarrow \Qnn$. Define
$\|\omega\|_1 := \sum_{g} \omega(g)$. An $m$-ary fractional operation $\omega$
is called an $m$-ary \emph{fractional polymorphism}~\cite{Cohen06:expressive} if
$\|\omega\|_1 = 1$ and for every function symbol $f \in \tau$ and tuples ${\bar
a_1}, \dots, {\bar a_m} \in D^{ar(f)}$, it holds that
\[ \sum_{g \in \opm{D}{m}} \omega(g) f^A(g({\bar a_1},\dots,{\bar a_m})) \leq
\frac{1}{m} \sum_{i=1}^m f^{A}({\bar a_i}). \]
The set $\{ g \mid \omega(g) > 0 \}$ of operations is called the \emph{support}
of $\omega$ and is denoted by $\supp(\omega)$.
Let $S_m$ be the symmetric group on $\{1,\dots,m\}$. An $m$-ary operation $g$ is
\emph{symmetric}\footnote{Symmetric operations are called totally
symmetric in~\cite{Kun12:itcs}.} if for every permutation $\pi \in S_m$, we
have
$ g(x_1,\dots,x_m) = g(x_{\pi(1)},\dots,x_{\pi(m)}). $

\begin{definition}\label{def:fracpolym}
  A \emph{totally symmetric fractional polymorphism} $\omega$ is a
  fractional polymorphism such that if $g \in \supp(\omega)$, 
  then $g$ is \emph{symmetric}.
\end{definition}

The \emph{superposition} of an $n$-ary operation $h$ with $n$ $m$-ary operations $g_1, \dots, g_n$ is the $m$-ary operation defined by $h[g_1,\dots,g_n](x_1,\dots,x_m) = h(g_1(x_1,\dots,x_m),\dots,g_n(x_1,\dots,x_m))$.
A set of operations %
is called a \emph{clone} if it contains all projections and 
is closed under superposition.
The smallest clone that contains a set of operations ${\cal F}$ is called the clone \emph{generated} by ${\cal F}$.
We say that an operation $f$ is \emph{generated} by ${\cal F}$ if it is contained in the clone generated by ${\cal F}$.

\begin{definition}
  The \emph{superposition}, $\omega[g_1,\dots,g_n]$, of an $n$-ary fractional
  polymorphism $\omega$ with $n$ $m$-ary operations $g_1,\dots,g_n$ is the $m$-ary
 fractional operation $\omega'$, where
  \[
  \omega'(h') = \sum_{h : h' = h[g_1,\dots,g_n]} \omega(h).
  \]
\end{definition}

Note that in general $\omega'$ is not a fractional polymorphism, but it does satisfy the following inequality:
\begin{eqnarray*}
\sum_{h' \in \opm{D}{m}} \omega'(h') f^A(h'({\bar a}_1,\dots,{\bar a}_m))
& = & \sum_{h \in \opm{D}{n}} \omega(h) f^A(h[g_1,\dots,g_n]({\bar a}_1,\dots,{\bar a}_m))\\
& \leq & 
\frac{1}{n} \sum_{i=1}^n f^A(g_i({\bar a}_1,\dots,{\bar a}_m)),
\end{eqnarray*}
for every $f \in \tau$ and ${\bar a}_1, \dots {\bar a}_m \in D^{ar(f)}$.

\subsection{The Multiset-Structure $P^m(A)$}

Let $A$ be a valued $\tau$-structure, $D = D(A)$, and let $m \geq 1$.
We define the \emph{multiset-structure}\footnote{A similar structure for $\{0,\infty\}$-valued languages was introduced in~\cite{Kun12:itcs}.} $P^m(A)$ as the valued structure with
domain $\multiset{D}{m}$, where $\multiset{D}{m}$ denotes the multisets of
elements from $D$ of size $m$, and for every $k$-ary function symbol $f \in \tau$, and $\alpha_1,\dots,\alpha_k \in \multiset{D}{m}$,
\[
f^{P^m(A)}(\alpha_1, \dots, \alpha_k) = \frac{1}{m} \min_{{\bar t_i} \in D^m : [{\bar t_i}] = \alpha_i} \sum_{i=1}^m f^{A}({\bar t_1}[i],\dots,{\bar t_k}[i]).
\]

The following lemma follows from the definitions
(the proof is in the appendix).

\begin{lemma} \label{lem:totsym-multiset}
  Let $A$ be a valued structure and $m > 1$.
  Then $P^m(A) \fhom A$ if and only if
  $A$ has an $m$-ary totally symmetric fractional polymorphism. 
\end{lemma}

\section{Basic Linear Programming Relaxation}\label{sec:blp}

Let $I$ and $A$ be valued structures over a common finite signature $\tau$.
Let $X = D(I)$ and $D = D(A)$.
The \emph{basic LP relaxation} (BLP) (sometimes also called the \emph{standard},
or \emph{canonical LP relaxation}) has variables $\lambda_{f,{\bar x},\sigma}$
for $f \in \tau$, ${\bar x} \in X^{ar(f)}$, $\sigma : \{{\bar x}\} \rightarrow
D$; and variables $\mu_{x}(a)$ for $x \in X, a \in D$.

\begin{equation}\label{eq:basiclp}
  \begin{array}{lll}
    \min
& \multicolumn{2}{l}{\displaystyle \sum_{f,{\bar x}} \sum_{\sigma : \{{\bar x}\} \rightarrow D} f^I({\bar x}) f^A(\sigma({\bar x})) \lambda_{f,{\bar x},\sigma}} \\
    \text{s.t.}
& \displaystyle\sum_{\sigma : \sigma(x) = a} \lambda_{f,{\bar x},\sigma} = \mu_{x}(a) 
& \qquad \text{ $\forall f \in \tau, {\bar x} \in X^{ar(f)}, x \in \{{\bar x}\}, a \in D$} \\
& \hspace*{0.8em} \displaystyle\sum_{a \in D} \mu_{x}(a) = 1 
& \qquad \text{ $\forall x \in X$} \\
\smallskip & 
\quad 0 \leq \lambda, \mu \leq 1 & \\
  \end{array}
\end{equation}

For any fixed $A$, BLP is polynomial in the size of a given
VCSP$(A)$ instance.
Let IP be the program obtained from~(\ref{eq:basiclp}) together with the
constraints that all variables take values in the range $\{0,1\}$ rather than
$[0,1]$. 
This is an integer programming formulation of the original VCSP instance. The
interpretation of the variables in IP is as follows: $\mu_{x}(a)=1$ iff variable
$x$ is assigned value $a$; $\lambda_{f,{\bar x},\sigma}=1$ iff constraint $f$ on
scope ${\bar x}$ is assigned tuple $\sigma({\bar x})$.
LP~(\ref{eq:basiclp}) is now a relaxation of IP and the question of whether
(\ref{eq:basiclp}) solves a given VCSP instance $I$ is the question of whether
IP has a zero integrality gap.
\section{Characterisation}\label{sec:char}

\begin{definition}
  Let BLP$(I,A)$ denote the optimum of $(\ref{eq:basiclp})$.
  We say that BLP \emph{solves} VCSP$(A)$ if BLP$(I,A) = \opt_A(I)$ for every instance $I$ of VCSP$(A)$.
\end{definition}

Solving the BLP provides an optimum value of the VCSP. 
To obtain an assignment achieving this value, we apply self-reduction:
Successively try each possible value for a variable and solve the altered LP.
Once the new optimum matches the original one, proceed with the next variable.

\begin{theorem}[Main]\label{thm:char}
  Let $A$ be a valued structure over a finite signature.
  TFAE:
  \begin{enumerate}[(i)]
  \item\label{main:1}
    BLP solves VCSP$(A)$.
  \item\label{main:2}
    For every $m>1$, $P^m(A) \fhom A$.
  \item\label{main:3}
    For every $m>1$, $A$ has an $m$-ary totally symmetric fractional
    polymorphism.
  \item\label{main:4}
    For every $n>1$, $A$ has a fractional polymorphism $\omega_n$ such that
    $\supp(\omega_n)$
    generates an $n$-ary symmetric operation.
  \end{enumerate}
\end{theorem}

The rest of this section is devoted to proving Theorem~\ref{thm:char}. 
We start with proving $(\ref{main:2})\Rightarrow(\ref{main:1})$.

\begin{theorem}\label{thm:forward}
  Assume that $P^m(A) \fhom A$ for every $m > 1$.
  Then BLP solves VCSP$(A)$.
\end{theorem}

\begin{proof}
  Let $\lambda^*, \mu^*$ be an optimal solution to (\ref{eq:basiclp}).
  Let $M$ be a positive integer such that $M \cdot \lambda^*$ and $M \cdot \mu^*$ are both integral.

  Let $\nu : X \rightarrow \multiset{D}{M}$ be defined by mapping $x$ to the multiset in which the elements are distributed according to $\mu^*_x$, i.e., the number of occurrences of $a$ in $\nu(x)$ is equal to $M \cdot \mu^*_{x}(a)$ for each $a \in D$. 
  
  Let $f$ be a $k$-ary function symbol in $\tau$, and let
  \[
  \sum_{{\bar x}, \sigma : \{{\bar x}\} \rightarrow D} f^I({\bar x}) f^A(\sigma({\bar x})) \lambda^*_{f,{\bar x},\sigma}
  = \sum_{{\bar x}} f^I({\bar x}) 
  \Big( 
  \sum_{\sigma : \{\bar{x}\} \rightarrow D} \lambda^*_{f,{\bar x},\sigma} f^A(\sigma({\bar x}))
  \Big)
  \]
  be the sum of all terms of the objective function in which $f$ occurs;
  Now, write
  \[
  M \cdot \sum_{\sigma : \{\bar{x}\} \rightarrow D} \lambda^*_{f,{\bar x},\sigma} f^A(\sigma({\bar x}))
  = f^A({\bar a_1}) + \dots + f^A({\bar a_M}),
  \]
  where the ${\bar a_i} \in D^k$ are such that a $\lambda^*_{f,{\bar x},\sigma}$-fraction are equal to $\sigma({\bar x})$. 

  Let ${\bar a_i}' = ({\bar a_1}[i],\dots,{\bar a_M}[i])$ for $i = 1, \dots, k$.
  \begin{eqnarray*}
  \sum_{\sigma : \{\bar{x}\} \rightarrow D} \lambda^*_{f,\bar{x},\sigma} f^A(\sigma({\bar x}))
  & = & \frac{1}{M} \sum_{i=1}^{M} f^A({\bar a_i}) 
  \ =\ \frac{1}{M} \sum_{i=1}^{M} f^A({\bar a_1}'[i],\dots,{\bar a_k}'[i]) \\
  & \geq &
  \frac{1}{M} \min_{{\bar t_i} \in D^M : [{\bar t_i}] = [{\bar a_i}']} \sum_{i=1}^M f^A({\bar t_1}[i],\dots,{\bar t_k}[i]) 
  \ =\ f^{P^M(A)}(\nu({\bar x})),
  \end{eqnarray*}
  where the last equality follows as the number of $a$'s in 
  $\bar{a}'_i$ is $M \cdot \sum_{\sigma : \sigma(\bar{x}[i]) = a} \lambda^*_{f,\bar{x},\sigma} = M \cdot \mu^*_{\bar{x}[i]}(a)$.

  We now have
  \begin{eqnarray*}
    BLP(I,A) & = &
    \sum_{f, {\bar x}} \sum_{\sigma : \{{\bar x}\} \rightarrow D} f^I({\bar x}) f^A(\sigma({\bar x})) \lambda^*_{f,{\bar x},\sigma} \\
    & = &
    \sum_{f \in \tau, {\bar x}} f^I({\bar x}) 
    \Big( 
    \sum_{\sigma : \{\bar{x}\} \rightarrow D} \lambda^*_{f,{\bar x},\sigma} f^A(\sigma({\bar x}))
    \Big) \\
    & \geq &
    \sum_{f \in \tau, {\bar x}} f^I({\bar x}) f^{P^M(A)}(\nu({\bar x})) \\
    & = & \opt_{P^M(A)}(I)
  \end{eqnarray*}
  
  It follows that $\opt_A(I) \geq BLP(I,A) \geq \opt_{P^M(A)}(I)$.
  Since $P^M(A) \fhom A$,
  the result then follows from Proposition~\ref{prop:frachom}.
\end{proof}

To prove $(\ref{main:1})\Rightarrow(\ref{main:2})$, we express the existence of
a fractional homomorphism $P^m(A) \fhom A$ as a system of linear inequalities.
We then apply a variant of Farkas' Lemma to show that if for some $m > 1$ there
is no such fractional homomorphism, then there exists an instance $I$ of
VCSP$(A)$ with a strictly greater optimum than BLP$(I,A)$ (the proof is in the
appendix).

\begin{theorem}\label{thm:back}
  Let $A$ be a valued structure and assume that BLP solves VCSP$(A)$.
  Then $P^m(A) \fhom A$ for every $m > 1$.
\end{theorem}

Lemma~\ref{lem:totsym-multiset} proves $(\ref{main:2})\Leftrightarrow
(\ref{main:3})$.

Since $(\ref{main:3})\Rightarrow (\ref{main:4})$ follows trivially, it remains to show that
$(\ref{main:4})\Rightarrow(\ref{main:3})$.

\begin{theorem}\label{thm:gentotot}
  Let $A$ be a valued structure and assume that for every $n>1$, $A$ has a fractional polymorphism $\omega_n$ that generates an $n$-ary symmetric operation. Then, for every $m > 1$, $A$ has an $m$-ary totally symmetric fractional polymorphism.
\end{theorem}

\begin{proof}
For an $m$-ary operation $g$, let $\tilde{g}$ denote the equivalence class
of $g$ under the relation:
\[
g \sim g' \Leftrightarrow g(x_1,\dots,x_m) = g'(x_{\pi(1)},\dots,x_{\pi(m)})
\text{ for some $\pi \in S_m$.}
\]
Note that we have $|\tilde{g}| = 1$ if and only if $g$ is symmetric.

We say that a fractional operation $\omega$ is \emph{weight-symmetric} if
\[
\omega(g) = \omega(g') \text{ whenever $g \sim g'$}.
\]

We construct an $m$-ary totally symmetric fractional polymorphism by building a
rooted tree in a number of stages. At each stage of the construction, every node
$u$ of the tree contains an $m$-ary weight-symmetric fractional operation with support on a single equivalence class of $\sim$.
For a node $u$, we will also denote this fractional operation by $u$.
Since $u$ is weight-symmetric, it follows that $u(g) = u(g')$ for all $g, g' \in \supp(u)$.
This common weight for the operations in the support of $u$ will be denoted by $w(u)$.
A node $u$ with $|\supp(u)| = 1$ will be called \emph{final}.

The following invariants are maintained throughout the construction.
\begin{enumerate}[(a)]
\item\label{inv:totsym}
  Every non-leaf node has at least one final child.
\item\label{inv:w}
  For every node $u$, we have $w(u) > 0$.
\item\label{inv:eq}
  For every non-leaf node $v$, 
  \[
  \sum_{\text{$u$ is a child of $v$}} \|u\|_1 = \|v\|_1.
  \]
\item\label{inv:ineq}
  For every non-leaf node $v$, every $f \in \tau$, and all tuples ${\bar a_1}, \dots, {\bar a_m} \in D^{ar(f)}$,
  \[
  \sum_{g} v(g) f^A(g({\bar a_1}, \dots, {\bar a_m}))
  \geq
  \sum_{\text{$u$ is a child of $v$}} \sum_{g} u(g) f^A(g({\bar a_1}, \dots, {\bar a_m})).
  \]
\end{enumerate}

We say that a leaf \emph{$u$ is covered (by $v$)},
and that $v$ is \emph{a covering node (of $u$)} if $\supp(u) = \supp(v)$ and
$u$ is a (proper) descendant of $v$. We say that $v$ is a \emph{minimal} covering node of
$u$ if no descendant of $v$ is a covering node of $u$.

At the beginning of the construction,
the  tree consists of a single root $r$ with $\supp(r)$ 
being the set of $m$-ary projections and $w(r) = \frac{1}{m}$.
We then apply the following two steps:
\begin{itemize}
\item
  \emph{Expansion:} A leaf $u$ that is not final and not covered is chosen to be expanded.
  This amounts to adding a finite non-empty set of children to $u$ while maintaining the invariants.
  The expansion step is repeated until no longer applicable.
\item
  \emph{Pruning:} A leaf that is not final and covered is removed together with a number of internal nodes while maintaining the invariants.
  The pruning step is repeated until no longer applicable.
\end{itemize}

Since there is a finite number of $m$-ary operations, and hence a finite number of equivalence classes of $\sim$, it follows that, eventually, every leaf in the tree that is not final must be covered. Hence, the expansion step is only applicable a finite number of times.

Each round of pruning shrinks the tree by at least one node,
but no final leaf is ever removed.
Therefore we eventually obtain a tree containing only final leaves,
at which time the pruning step is no longer applicable.
Let $\mathcal{L}$ be the set of leaves in the final tree.
By repeated application of invariant (\ref{inv:ineq}), starting from the root,
$\sum_{u \in \mathcal{L}} u$ is then an $m$-ary totally symmetric fractional polymorphism.

\medskip
\paragraph{\emph{Expansion.}}
We expand a leaf $u$ with $|\supp(u)| = n$ as follows:
Let $\omega$ be a $k$-ary fractional polymorphism of $A$ such that
$\supp(\omega)$ generates an $n$-ary symmetric operation $t$.

We will define a sequence of $m$-ary weight-symmetric fractional operations $\nu_i$,
each with $\|\nu_i\|_1 = \|u\|_1$.
Let $\nu_0 = u$.
Assume that $\nu_{i-1}$ has been defined for some $i \geq 1$.
Let $l_{i-1} = \min \{ \nu_{i-1}(g) \mid g \in \supp(\nu_{i-1}) \}$ be the
minimum weight of an operation in the support of $\nu_{i-1}$.
The fractional operation $\nu_i$ is obtained by subtracting from $\nu_{i-1}$
an equal amount of weight from each operation in $\supp(\nu_{i-1})$ and
adding this weight as superpositions of $\omega$ by all possible
choices of operations in $\nu_{i-1}$.
The amount subtracted from each operation is $\frac{1}{2} l_{i-1}$
so that every operation in $\supp(\nu_{i-1})$ is also in $\supp(\nu_i)$.
Formally $\nu_i$ is defined as follows:
\[ %
\nu_i = 
\nu_{i-1} - 
\frac{1}{2} l_{i-1} \chi_{i-1} + 
\sum_{(g_1,\dots,g_k) \in {\rm supp}(\nu_{i-1})^k} \frac{1}{2} l_{i-1} \frac{1}{K} \omega[g_1,\dots,g_k],
\] %
where $K = |\supp(\nu_{i-1})|^k$ and $\chi_{i-1}$ is the indicator function of $\supp(\nu_{i-1})$.

By definition $\|\nu_i\|_1 = \|\nu_{i-1}\|_1 = \|u\|_1$.
To verify that $\nu_i$ is weight-symmetric, it suffices to verify that the sum
\begin{equation} \label{eqn:sym}
\sum_{(g_1,\dots,g_k) \in {\rm supp}(\nu_{i-1})^k} \omega[g_1,\dots,g_k]
\end{equation}
is weight-symmetric.
Let $g \sim g'$, let $\pi \in S_m$ be such that $g(x_1,\dots,x_m) =
g'(x_{\pi(1)},\dots,x_{\pi(m)})$ and let $g'_j(x_1,\dots,x_m) =
g_j(x_{\pi(1)},\dots,x_{\pi(m)})$ for $1 \leq j \leq m$.
Since $\nu_{i-1}$ is weight-symmetric, it follows that
$g_i \in {\rm supp}(\nu_{i-1})$ if and only if $g'_i \in {\rm supp}(\nu_{i-1})$.
Therefore the terms $\omega(h) h[g_1,\dots,g_k]$ in (\ref{eqn:sym}) such that $g = h[g_1,\dots,g_k]$ are in bijection with the terms
$\omega(h) h[g'_1,\dots,g'_k]$ such that $g' = h[g'_1,\dots,g'_k]$.
So the fractional operation in (\ref{eqn:sym}) assigns the same weight to $g$ and $g'$.

Let $e$ be an expression for $t$ consisting of superpositions of projections and operations from $\supp(\omega)$.
We recursively define the \emph{nested depth}, $d = d(e)$, of $e$ as follows:
$d(p) = 0$ for every projection $p$; and
$d(h[g_1,\dots,g_k]) = 1+\max_{1\leq i \leq k} d(g_i)$.

Let $\supp(u) = \{g_1,\dots,g_n\}$.
Using $\supp(\nu_0) = \supp(u)$
and the fact that $\supp(\nu_i)$ contains all superpositions of operations
in $\supp(\nu_{i-1})$,
it follows that $t[g_1,\dots,g_n] \in \supp(\nu_d)$.
Now, we add a child $v$ to $u$ for every equivalence class in the set
$\{ \tilde{g} \mid g \in \supp(\nu_d) \}$.
For an added child $v$ with $\supp(v) = \tilde{g}$, 
we let $w(v) = \nu_d(g)$.

Invariant (\ref{inv:totsym}) holds as $t[g_1,\dots,g_n] \in \supp(\nu_d)$ is symmetric:
for all $\pi \in S_m$ there is a $\pi' \in S_n$ such that
$t[g_1,\dots,g_n](x_{\pi(1)},\dots,x_{\pi(m)}) = t[g_{\pi'(1)},\dots,g_{\pi'(n)}](x_1,\dots,x_m) = t[g_1,\dots,g_n](x_1,\dots,x_m)$.
Invariants (\ref{inv:w}) and (\ref{inv:eq}) hold by construction.
For each $i \geq 1$, we have
\begin{eqnarray*}
\sum_{g} \nu_{i-1}(g) f^A(g({\bar a}_1,\dots,{\bar a}_m)) \geq
\sum_{g} \nu_{i}(g) f^A(g({\bar a}_1,\dots,{\bar a}_m))
\end{eqnarray*}
for all $f \in \tau$ and ${\bar a}_1, \dots {\bar a}_m \in D^{ar(f)}$.
Therefore invariant (\ref{inv:ineq}) also holds after expanding $u$.

\medskip
\paragraph{\emph{Pruning.}}
The pruning step maintains an additional invariant, namely that every
leaf that is not final is covered.
Pruning is accomplished as follows.
Pick a minimal covering node $v$.
Let $\nu = \nu_v + \nu_\bot$ be the fractional operation induced by the leaves in the subtree rooted at $v$, where $\nu_v$ is the part of $\nu$ with the same support as $v$ and $\nu_\bot$ is the part of $\nu$ with support disjoint from $v$.
Inductively, by invariant (\ref{inv:ineq}),
\[
  \sum_{g} v(g) f^A(g({\bar a_1}, \dots, {\bar a_m})) \geq
  \sum_{g} \nu_v(g) f^A(g({\bar a_1}, \dots, {\bar a_m})) +
  \sum_{g} \nu_\bot(g) f^A(g({\bar a_1}, \dots, {\bar a_m})),
\]
for all $f \in \tau$ and ${\bar a}_1, \dots {\bar a}_m \in D^{ar(f)}$.
We simplify this inequality as follows.
\[
  \sum_{g} v(g) f^A(g({\bar a_1}, \dots, {\bar a_m})) \geq
  \sum_{g} \frac{1}{1-\kappa} \nu_\bot(g) f^A(g({\bar a_1}, \dots, {\bar a_m})),
\]
where $\kappa = \|\nu_v\|_1/\|v\|_1$.

Remove all nodes below $v$ and
add a new child $u$ to $v$ for every equivalence class in the set
$\{ \tilde{g} \mid g \in \supp(\nu_\bot) \}$.
For an added child $u$ with $\supp(u) = \tilde{g}$, 
we let $w(u) = \frac{1}{1-\kappa} \nu_\bot(g)$.

By invariant (\ref{inv:totsym}) the node $v$ is guaranteed to have at least one
final child (leaf).
Hence, by invariant (\ref{inv:w}), $\nu_\bot$ is not identically 0.
By induction on (\ref{inv:eq}), it follows that $\|v\|_1 > \|\nu_v\|_1$, so $\kappa < 1$ and
the new weights are defined and positive.
So invariant (\ref{inv:w}) holds.
Furthermore,
\[
\sum_{\text{$u$ is a child of $v$}} \|u\|_1 =
\frac{1}{1-\kappa} \|\nu_\bot\|_1 =
\frac{1}{1-\kappa} (\|\nu\|_1 - \|\nu_v\|_1) =
\frac{1}{1-\kappa} (\|v\|_1 - \|\nu_v\|_1) =
\|v\|_1,
\]
so invariant (\ref{inv:eq}) holds.
Invariant (\ref{inv:ineq}) holds by construction.
Finally, invariant (\ref{inv:totsym}) 
also holds since for any final child $u$ of $v$, there is still a child of $v$
with support on the same equivalence class as $u$.

Only nodes in the subtree rooted at $v$ have been removed and every child
added to $v$ has the same support as a previous leaf of this subtree.
Every such leaf $u$ that was not final was covered by a
node above $v$ in the tree. Hence, all leaves in the tree that are
not final are still covered.
\end{proof}

\section{General-valued Structures} \label{sec:gener-valued}

The valued structures we have dealt with so far were in fact \emph{finite-valued};
i.e., for each function symbol $f\in\tau$, the range of $f^A$ was $\Qnn$. We
now discuss how our result can be extended to the \emph{general-valued} case,
in which for each function symbol $f\in\tau$, the range of $f^A$ is
$\Qnnc=\Qnn\cup\{\infty\}$. (We define $c+\infty=\infty+c=\infty$ for
all $c\in\Qnnc$, and $0\infty=\infty0=0$.)

Inspired by the OSAC algorithm~\cite{Cooper10:osac}, the algorithm for
general-valued structures, denoted by BLP$_g$, works in two stages. Firstly, the
instance is made arc consistent using a standard  arc consistency
algorithm~\cite{Freuder78:synthesizing}.\footnote{The algorithm
from~\cite{Freuder78:synthesizing} is sometimes called generalised arc
consistency algorithm to emphasise the fact that it works for CSPs of arbitrary
arities, not only for binary CSPs~\cite{Mackworth77:consistency}.}
Secondly, BLP (i.e., linear program~(1) from Section~\ref{sec:blp}) is solved.

\begin{definition}\label{def:polym}
Let $A$ be a general-valued $\tau$-structure, and let $D=D(A)$. An $m$-ary operation
$g:D^m\rightarrow D$ is a \emph{polymorphism} of 
$A$ if for every function
symbol $f\in\tau$ and tuples ${\bar a_1},\ldots,{\bar a_m}\in D^{ar(f)}$, it holds that 
$
\feas(f^A(g({\bar a_1},\dots,{\bar a_m}))) \leq \sum_{i=1}^m \feas(f^{A}({\bar a_i})),
$
where $\feas(\infty)=\infty$ and $\feas(c)=0$ for any $c\in\Qnn$.
\end{definition}

From Definitions~\ref{def:fracpolym} and~\ref{def:polym}, if $\omega$ is a
fractional polymorphism of a general-valued structure $A$, then every
$g\in\supp(\omega)$ is a polymorphism of $A$. In fact, any operation $g$ that is
generated by $\supp(\omega)$ (i.e., $g$ belongs to the clone generated by $\supp(\omega)$) is a
polymorphism of $A$. (For finite-valued structures, trivially, every operation $g$ is a
polymorphism.)

Arc consistency (also known as
$(1,k)$-consistency)~\cite{Freuder78:synthesizing} is a decision procedure for
precisely those $\{0,\infty\}$-valued structures $A$ that are closed under a
\emph{set function} $g : 2^{D(A)} \setminus \{ \emptyset \} \rightarrow D(A)$~\cite{Dalmau99:set,Feder98:monotone}.
This condition has recently been shown to be equivalent to the requirement
that $A$ should have symmetric polymorphisms of all arities~\cite{Kun12:itcs}.

Our main theorem (Theorem~\ref{thm:char}) also holds for general-valued structures:

\begin{theorem}\label{thm:char2}
  Let $A$ be a general-valued structure over a finite signature. 
  TFAE:
  \begin{enumerate}[(i)]
  \item\label{gmain:1}
    BLP$_g$ solves VCSP$(A)$.
  \item\label{gmain:2}
    For every $m>1$, %
    $P^m(A) \fhom A$.
  \item\label{gmain:3}
    For every $m>1$, $A$ has an $m$-ary 
    totally symmetric fractional polymorphism.
  \item\label{gmain:4}
    For every $n>1$, $A$ has 
    a fractional polymorphism $\omega_n$ such that $\supp(\omega_n)$ generates an $n$-ary symmetric operation.

  \end{enumerate}
\end{theorem}

\begin{proof}

Note that $(\ref{gmain:2})$ implies that $A$ has symmetric polymorphisms of all
arities. This follows from Lemma~\ref{lem:totsym-multiset}, which holds for
general-valued structures, and the above-mentioned fact that any
$g\in\supp(\omega)$, where $\omega$ is a fractional polymorphism of $A$, is a
polymorphism of $A$. The same argument guarantees symmetric polymorphisms of all
arities in $(\ref{gmain:3})$ and $(\ref{gmain:4})$; in $(\ref{gmain:4})$, we use
that fact than any operation generated by $\supp(\omega_n)$ is a polymorphism of
$A$.

Theorem~\ref{thm:forward} proves $(\ref{gmain:2})\Rightarrow(\ref{gmain:1})$
since the assumption of having symmetric polymorphisms of all arities guarantees
a feasible solution to~(\ref{eq:basiclp}). From the discussion above on arc
consistency, if BLP$_g$ solves VCSP$(A)$, then $A$ has symmetric polymorphisms
of all arities since arc consistency decides the existence of a finite-valued
solution. Furthermore, the proof of Theorem~\ref{thm:back} shows that having
symmetric polymorphisms of all arities but not having a fractional homomorphism
$P^m(A) \fhom A$ implies that BLP$_g$ does not solve VCSP$(A)$. This gives
$(\ref{gmain:1})\Rightarrow (\ref{gmain:2})$. $(\ref{gmain:2})\Leftrightarrow
(\ref{gmain:3})$ and $(\ref{gmain:3})\Leftrightarrow(\ref{gmain:4})$ are proved
the same way as in Theorem~\ref{thm:char}, by Lemma~\ref{lem:totsym-multiset}
and Theorem~\ref{thm:gentotot}, respectively.
\end{proof}

\begin{remark}
BLP$_g$ for general-valued structures uses the arc consistency
algorithm~\cite{Freuder78:synthesizing} and BLP. Since Kun et
al~\cite{Kun12:itcs} have shown that BLP solves CSPs (i.e., $\{0,\infty\}$-valued
VCSPs), if represented by $\{0,1\}$-valued structures, of width 1 -- thus
providing an alternative to the standard arc consistency
algorithm~\cite{Freuder78:synthesizing} -- a different approach is to combine
Theorem~\ref{thm:char} from this paper with ~\cite{Kun12:itcs} and
solve general-valued structures using only BLP with an amended objective
function which takes care of infinite costs using a large (but polynomial),
instance-dependent constant. 
\end{remark}

\section{Tractable Valued Constraint Languages}\label{sec:trac}

As before, we denote $D=D(A)$. A \emph{binary
multimorphism}~\cite{Cohen06:complexitysoft} of a valued structure $A$ is a pair
$\mm{g_1}{g_2}$ of binary functions $g_1,g_2:D^2\rightarrow D$ such that for
every function symbol $f \in \tau$ and tuples $\bar{a}_1,\bar{a}_2 \in D^{ar(f)}$, 
it holds that
\[
f^A(g_1(\bar{a}_1,\bar{a}_2))+f^A(g_2(\bar{a}_1,\bar{a}_2)) \leq f^A(\bar{a}_1)+f^A(\bar{a}_2). 
\]
(Multimorphisms are a special case of fractional polymorphisms.)
Since any semi-lattice
operation\footnote{A semi-lattice operation  is associative,
commutative, and idempotent.} generates symmetric operations of all
arities, we get: 
\begin{corollary}[of Theorem~\ref{thm:char2}]
Let $A$ be a valued structure with a binary multimorphism $\mm{g_1}{g_2}$ where
either $g_1$ or $g_2$ is a semi-lattice operation. Then $A$ is tractable.
\end{corollary}

We now give examples of valued structures (i.e., valued constraint languages)
defined by such binary multimorphisms. 

\begin{example} Let $(D; \meet, \join)$ be an \emph{arbitrary} lattice on $D$.
Assume that a valued structure $A$ has the multimorphism $\mm{\meet}{\join}$.
Then VCSP$(A)$ is tractable. The tractability of $A$ was previously known only
for distributive lattices~\cite{Schrijver00:submodular,Iwata01:submodular} and
(finite-valued) diamonds~\cite{Kuivinen11:do-diamonds}, see
also~\cite{Krokhin08:max}. \end{example}

\begin{example} A pair of operations $\mm{g_1}{g_2}$ is called a symmetric
tournament pair (STP) if both $g_1$ and $g_2$ are commutative,
conservative ($g_1(x,y)\in\{x,y\}$ and $g_2(x,y)\in \{x,y\}$ for all $x,y\in
D$), and $g_1(x,y)\neq g_2(x,y)$ for all $x,y\in D$.
Let $A$ be a finite-valued structure with an STP multimorphism $\mm{g_1}{g_2}$.
It is known that if a finite-valued structure admits an STP multimorphism, it
also admits a submodularity multimorphism. This result is implicitly contained
in~\cite{Cohen08:Generalising}.\footnote{Namely, the STP might contain cycles,
but \cite[Lemma 7.15]{Cohen08:Generalising} tells us that on cycles we have, in
the finite-valued case, only unary cost functions. It follows that the cost
functions admitting the STP must be submodular with respect to some total
order.} Therefore, BLP solves any instance from VCSP$(A)$. \end{example}

\begin{example} \label{ex:bisub} Assume that a valued structure $A$ is
bisubmodular~\cite{Fujishige06:bisubmodular}. This means that $D=\{0,1,2\}$ and
$A$ has a multimorphism $\mm{\min_0}{\max_0}$~\cite{Cohen06:complexitysoft},
where $\min_0(x,x)=x$ for all $x\in D$ and $\min_0(x,y)=0$ for all $x,y\in D,
x\neq y$; $\max_0(x,y)=0$ if $0\neq x\neq y\neq 0$ and $\max_0(x,y)=\max(x,y)$
otherwise, where $\max$ returns the larger of its two arguments with respect to
the normal order of integers. Since $\min_0$ is a semi-lattice operation, $A$ is
tractable. The tractability of (finite-valued) $A$ was previously known only
using a general algorithm for bisubmodular functions given by an
oracle~\cite{Fujishige06:bisubmodular,McCormick10:bisubmodular}. \end{example}

\begin{example} Assume that a valued structure $A$ is 
\emph{weakly tree-submodular} on
an \emph{arbitrary} tree~\cite{Kolmogorov11:mfcs}. The meet (which is defined as
the highest common ancestor) is again a semi-lattice operation. The same holds
for \emph{strongly tree-submodular} structures since strong tree-submodularity
implies weak tree-submodularity~\cite{Kolmogorov11:mfcs}. 
The tractability of weakly
tree-submodular valued structures was previously known only for chains and
forks~\cite{Kolmogorov11:mfcs}. The tractability of strongly tree-submodular
valued structures was previously known only for binary
trees~\cite{Kolmogorov11:mfcs}. \end{example}

\begin{example}
  Note that the previous example applies to \emph{all} trees, not just binary
  ones. In particular, it applies to the tree consisting of one root with $k$
  children. This is equivalent to structures with $D=\{0,1,\dots,k\}$ and the
  multimorphism $\mm{\min_0}{\max_0}$ from Example~\ref{ex:bisub}.
  This is a natural generalisation of submodular ($k=1$) and bisubmodular
  ($k=2$) functions, known as $k$-submodular functions~\cite{Huber12:ksub}. The
  tractability of $k$-submodular valued structures for $k>2$ was previously
  open.
\end{example}

\begin{example}

Let $b$ and $c$ be two distinct elements of $D$ and let $(D;<)$ be a partial order
which relates all pairs of elements except for $b$ and $c$. A pair 
$\tuple{g_1,g_2}$, where $g_1,g_2:D^2\rightarrow D$ are two binary operations, is a
\emph{1-defect} multimorphism if $g_1$ and $g_2$ are both commutative and satisfy the following
conditions:
\begin{itemize}
\item If $\{x,y\}\neq\{b,c\}$, then $g_1(x,y)=x\meet y$ and $g_2(x,y)=x\join y$.
\item If $\{x,y\}=\{b,c\}$, then $\{g_1(x,y),g_2(x,y)\}\cap\{x,y\}=\emptyset$, and
$g_1(x,y)<g_2(x,y)$.
\end{itemize}

The tractability of valued structures that have a 1-defect multimorphism has
recently been shown in~\cite{Jonsson11:cp}. We now show that valued structures 
with a 1-defect multimorphism are solvable by BLP$_g$.

Without loss of generality, we assume that $g_1(b,c)<b,c$ and write $g=g_1$.
(Otherwise, $g_2(b,c)>b,c$, and $g_2$ is used instead.)
Using $g$, we construct a symmetric $m$-ary operation $f(x_1,\dots,x_m)$.

Let $f_1,\ldots,f_M$ be the $M={m\choose 2}$ terms $g(x_i,x_j)$.
Let $f=g(f_1,g(f_2,\dots,g(f_{M-1},f_M)\dots))$. There are three possible cases:

\begin{itemize}

\item 

$\{b,c\} \not\subseteq {x_1,...,x_m}$. Then $g$ acts as $\meet$, which is a
semi-lattice operation, hence so does $f$.

\item

$\{b,c\} \subseteq \{x_1,...,x_m\}$ and $g(b,c)\leq x_1,\dots,x_m$. Then
$f_i=g(b,c)$ for some $1\leq i\leq M$, and $g(f_i,f_j)=g(b,c)$ for all $1\leq
j\leq M$, so $f(x_1,\dots,x_m)=g(b,c)$.

\item

$\{b,c\} \subseteq \{x_1,\dots,x_m\}$ and there is a variable $x_p$ for some $1\leq p\leq
m$ such that $x_p\leq
g(b,c)$ and $x_p\leq x_1,\dots,x_m$. Then
$g(x_p,x_q)=x_p$ for all $1\leq q\leq m$ so $f_i=x_p$ for some $1\leq i\leq M$ and $g(f_i,f_j)=x_p$
for all $1\leq j\leq M$, so $f(x_1,\dots,x_m)=x_p$.

\end{itemize}

\end{example}

\section{Conclusions}

We have characterised precisely for which valued structures the basic linear
programming relaxation (BLP) is a decision procedure. This implies tractability
of several previously open classes of VCSPs including several generalisations of
submodularity. In fact, BLP solves \emph{all known} tractable finite-valued
structures. 

The main result does not give a decidability criterion for testing whether a
valued structure is solvable by BLP. Interestingly, all known tractable
finite-valued structures have a binary multimorphism. It is possible that
every (finite-)valued structure solvable by BLP admits a fixed-arity
multimorphism, which would give a polynomial-time checkable condition.

An intriguing open question is whether our tractability results hold in the
oracle-value model; that is, for objective functions which are not given
explicitly as a sum of cost functions, but only by an oracle. For instance, the
\emph{maximisation problem}
for submodular functions on distributive lattices, known to
be NP-complete, allows for good approximation algorithms in both
models~\cite{Vondrak11:max-non-monotone}.

\newpage

\bibliographystyle{plain}
\bibliography{tz12arxiv-v2}

\newpage
\appendix

\section{Proof of Lemma~\ref{lem:totsym-multiset}}

\begin{lemma} %
  Let $A$ be a valued structure and $m > 1$.
  Then $P^m(A) \fhom A$ if and only if
  $A$ has an $m$-ary totally symmetric fractional polymorphism. 
\end{lemma}

\begin{proof}
  Let $D = D(A)$.

  $(\Rightarrow)$\
  Let $\omega$ be a fractional homomorphism from $P^m(A)$ to $A$.
  Let $h : D^m \rightarrow \multiset{D}{m}$ be the function that sends 
  a tuple $(a_1,\dots,a_m) \in D^m$ to its multiset $[a_1,\dots,a_m]$.
  Then the composition $\omega \circ h$ is the desired fractional polymorphism.

  $(\Leftarrow)$\
  Let $\omega$ be an $m$-ary totally symmetric fractional polymorphism.
  Every operation $g \in \supp(\omega)$ induces a function 
  $g' : \multiset{D}{m} \rightarrow D$.
  Define $\omega'$ as the fractional homomorphism with
  $\omega'(g') = \omega(g)$.
  Let $f$ be a $k$-ary function symbol in the signature of $A$,
  let $\alpha_1, \dots, \alpha_k \in \multiset{D}{m}$ be arbitrary, and
  pick ${\bar a_i}$ with $[\bar{a}_i] = \alpha_i$ that minimises $\sum_{i=1}^m f^A({\bar a_1}[i],\dots,{\bar a_k}[i])$, for $1 \leq i \leq k$.
  \begin{eqnarray*}
    \sum_{g'} \omega'(g') f^{A}(g'(\alpha_1),\dots,g'(\alpha_k))
    & = & \sum_{g} \omega(g) f^{A}(g({\bar a_1}),\dots,g({\bar a_k}))\\
    & \leq & \frac{1}{m} \sum_{i=1}^m f^{A}(a_1[i],\dots,a_k[i])\\
    & = & \frac{1}{m} \min_{{\bar t_i} \in D^m : [{\bar t_i}] = \alpha_i} \sum_{i=1}^m f^A({\bar t_1}[i],\dots,{\bar t_k}[i])\\
    & = & \frac{1}{m} \sum_{i=1}^m f^{P^m(A)}(\alpha_1,\dots,\alpha_k).
  \end{eqnarray*}
  Hence, $P^m(A) \fhom A$.
\end{proof}

\section{Proof of Theorem~\ref{thm:back}}

The following variant of Farkas' Lemma is due to Gale~\cite{gale60:economicmodels} (cf. Mangasarian~\cite{Mangasarian94:nonlinear}). 

\begin{lemma}\label{lem:gale}
  Let $A \in \R^{m \times n}$ and ${\bar b} \in \R^m$.
  Then exactly one of the two holds: %
  \begin{itemize}
  \item
    $A{\bar x} \leq {\bar b}$ for some ${\bar x} \in \R^n$; or
  \item
    $A^T{\bar y} = 0$, ${\bar b}^T{\bar y} = -1$ for some ${\bar y} \in \R_{\geq 0}$.
  \end{itemize}
\end{lemma}

\begin{theorem}
  Let $A$ be a valued structure and assume that BLP solves VCSP$(A)$.
  Then $P^m(A) \fhom A$ for every $m > 1$.
\end{theorem}

\begin{proof}
  Let $\tau$ be the signature of $A$ and let $D = D(A)$.
  We prove the contrapositive.
  Assume that there is an integer $m > 1$ such that $P^m(A)$ does not have a fractional homomorphism to $A$.
  Let $\Omega$ denote the set of functions from $\multiset{D}{m}$ to $D$.
  We are assuming that the following system of inequalities does not have a solution $\omega : \Omega \rightarrow \Qnn$.
  \begin{eqnarray*}
    \sum_{g \in \Omega} \omega(g) f^A(g(\bar{\alpha})) & \leq &
    f^{P^m(A)}(\bar{\alpha})
    \quad \text{ $\forall f \in \tau, \bar{\alpha} \in \textstyle \multiset{D}{m}^{ar(f)}$}\\ %
    \sum_{g \in \Omega} \omega(g) & = & 1 \\
    \omega(g) & \geq & 0 \quad \text{ $\forall g \in \Omega$.}
  \end{eqnarray*}

  In order to apply Lemma~\ref{lem:gale}, we rewrite the equality $\sum_{g} \omega(g) = 1$ into two inequalities $\sum_{g} \omega(g) \leq 1$ and $-\sum_{g} \omega(g) \leq -1$.
  The last set of inequalities are rewritten to the form $-\omega(g) \leq 0$ for each $g \in \Omega$.
  We have one variable for each inequality, i.e., $y(f,\bar{\alpha})$ for $f \in \tau$, and $\bar{\alpha} \in \multiset{D}{m}^{ar(f)}$.
  Additionally, we have two variables $z_+, z_-$ for the two inequalities involving the constant 1 and one variable $w(g)$ for each $g \in \Omega$.
  \begin{eqnarray*}
    \sum_{f,\bar{\alpha}} y(f,\bar{\alpha}) f^A(g(\bar{\alpha})) + z_+ - z_- - w(g) & = & 0\\
    \sum_{f,\bar{\alpha}} y(f,\bar{\alpha}) f^{P^m(A)}(\bar{\alpha}) + z_+ - z_- & = & -1 \\
    y, z_+, z_-, w & \geq & 0
  \end{eqnarray*}
  
  We can isolate $z_++z_-$ in the last equality, 
  \[
  z_++z_- = -1 - \sum_{f,\bar{\alpha}} y(f,\bar{\alpha}) f^{P^m(A)}(\bar{\alpha}),
  \]
  which substituted into the first set of equalities implies that there is a solution $y(f,\bar{\alpha}), w(g)$ such that, for each $g \in \Omega$,
  \[
  \sum_{f,\bar{\alpha}} y(f,\bar{\alpha}) f^A(g(\bar{\alpha})) = w(g) + 1 + \sum_{f,\bar{\alpha}} y(f,\bar{\alpha}) f^{P^m(A)}(\bar{\alpha}).
  \]

  We therefore find that there is a solution to the following system:
  \begin{eqnarray}\label{eq:fracsol}
    \sum_{f,\bar{\alpha}} y(f,\bar{\alpha}) f^A(g(\bar{\alpha})) & > & \sum_{f,\bar{\alpha}} y(f,\bar{\alpha}) f^{P^m(A)}(\bar{\alpha}) 
    \quad \text{ $\forall g \in \Omega$} \notag \\
    y(f,\bar{\alpha}) & \geq & 0 \quad \text{ $\forall f, \bar{\alpha}$.} \notag
  \end{eqnarray}

  Let $I$ be the instance on variables $\multiset{D}{m}$.
  For each $k$-ary function symbol $f \in \tau$, and $\bar{\alpha} \in \multiset{D}{m}^{ar(f)}$, define
  \[
  f^I(\bar{\alpha}) = y(f,\bar{\alpha}).
  \]

  We now give a solution $\lambda, \mu$ to the basic LP~(\ref{eq:basiclp}) with an objective value equal to the right-hand side of (\ref{eq:fracsol}).
  Each variable $\mu_{\alpha}(a)$ is assigned the value of the multiplicity of $a$ in $\alpha$ divided by $m$.
  Given $f, \bar{\alpha}$,
  let ${\bar t_1},\dots,{\bar t_k} \in D^m$ be such that $f^{P^m(A)}(\bar{\alpha}) = \frac{1}{m} \sum_{i=1}^m f^A(\bar{t}_1[i],\dots,\bar{t}_k[i])$,
  and assign values to the $\lambda$-variables as follows:
  \[
  \lambda_{f,\bar{\alpha},\sigma} =
  \frac{1}{m} |\{ i \mid \sigma(\bar{\alpha}[j]) = \bar{t}_j[i] \text{ for all $j$} \}|
  \]

  Note that $\sum_{\sigma : \sigma(\bar{\alpha}[j]) = a} \lambda_{f,\bar{\alpha},\sigma} = \mu_{\bar{\alpha}[j]}(a)$ for all $1 \leq j \leq k$ and $a \in D$.
  Furthermore, $\lambda$ is defined so that
  $f^{P^{m}(A)}(\bar{\alpha}) = \sum_{\sigma : \{\bar{\alpha}\} \rightarrow D} f^A(\sigma(\bar{\alpha})) \lambda_{f,\bar{\alpha},\sigma}$.
  Hence, the variables $\lambda, \mu$ satisfy the basic LP~(\ref{eq:basiclp}),
  and we have
  \begin{equation}\label{eq:altfracsol}
    BLP(I,A) \leq 
    \sum_{f,\bar{\alpha}} f^I(\bar{\alpha}) \sum_{\sigma : \{\bar{\alpha}\} \rightarrow D} f^A(\sigma(\bar{\alpha})) \lambda_{f,\bar{\alpha},\sigma}
    = 
    \sum_{f,\bar{\alpha}} f^I(\bar{\alpha}) f^{P^m(A)}(\bar{\alpha}),
  \end{equation}
  where the sums are over $f \in \tau$ and $\bar{\alpha} \in \multiset{D}{m}^{ar(f)}$.

  It now follows from (\ref{eq:fracsol}) and (\ref{eq:altfracsol}) that the measure of any solution $g : \multiset{D}{m} \rightarrow D$ to $I$ is strictly greater than BLP$(I,A)$.
  Consequently,
  BLP does not solve VCSP$(A)$.
\end{proof}

\section{Optimal Soft Arc Consistency}\label{sec:osac}

In this section we define optimal soft arc consistency, which is closely related
to BLP given by~(\ref{eq:basiclp}) in
Section~\ref{sec:blp}.

Let $I$ and $A$ be valued structures over a common finite signature.
Let $X = D(I)$ and $D = D(A)$.
We will group the terms of an instance with respect to their \emph{scope}.
Let $S \subseteq X$.
The terms of this scope are those of the form $f^I({\bar x})f^A(\sigma({\bar x}))$,
where $\{{\bar x}\} = S$,
and $\text{ar}(f) = |{\bar x}|$.
For each scope $S$, $x \in S$, and $\sigma : S \rightarrow D$, we have a variable $y_{S,x}(\sigma(x))$.
For each $x \in X$, we have a variable $z_x$.

Establishing \emph{optimal soft arc consistency} (OSAC) amounts to solving the
following linear program~\cite{Cooper10:osac}:
\begin{equation}\label{eq:lposac}
  \begin{array}{lll}
    \max
& \displaystyle \sum_{x} z_x & \\
    \text{s.t.}
&
    \displaystyle \sum_{\{{\bar x}\} = S, f} f^I({\bar x}) f^A(\sigma({\bar x})) - \displaystyle \sum_{x \in S} y_{S,x}(\sigma(x)) \geq 0 
&
\qquad   \text{ $\forall S \subseteq X, \sigma : S \rightarrow D$} \\
&
\hspace*{0.4em}\displaystyle \sum_{u} u^I(x) u^A(\sigma(x)) - z_x + \displaystyle \sum_{S : x \in S} y_{S,x}(\sigma(x)) \geq 0
& 
    \qquad \text{ $\forall x \in X, \sigma : \{x\} \rightarrow D$} \\
  \end{array}
\end{equation}

We refer the reader to~\cite{Cooper10:osac} for more details, but the idea
behind~(\ref{eq:lposac}) is that it gives the maximum lower bound on $\opt_A(I)$
among all arc-consistency closures of the given instance $I$, where the closure
is obtained by repeated calls of three basic operations called Extend, Project,
and  UnaryProject.

We will be interested in the dual of (\ref{eq:lposac}).
The dual has variables $\lambda_{S,\sigma}$ for $S \subseteq X$ and $\sigma : S \rightarrow D$,
and variables $\mu_{x}(a)$ for $x \in X, a \in D$.

\begin{equation}\label{eq:dual}
  \begin{array}{lll}
    \min
& \multicolumn{2}{l}{\displaystyle \sum_{S \subseteq X,\sigma} \Big( \sum_{\{{\bar x}\} = S, f} f^I({\bar x}) f^A(\sigma({\bar x})) \Big) \lambda_{S,\sigma} +
  \displaystyle \sum_{x \in X, \sigma} \Big( \sum_{u} u^I(x) u^A(\sigma(x))
  \Big) \mu_{x}(\sigma(x))}  \\
    \text{s.t.}
& \displaystyle\sum_{\sigma : \sigma(x) = a} \lambda_{S,\sigma} = \mu_{x}(a) 
& \qquad \text{ $\forall S \subseteq X, x \in S, a \in D$} \\
& \hspace*{0.8em} \displaystyle\sum_{a \in D} \mu_{x}(a) = 1 
& \qquad \text{ $\forall x \in X$} \\
\smallskip & 
\quad\lambda, \mu \geq 0 & \\
  \end{array}
\end{equation}

Note that~(\ref{eq:dual}) is a tighter relaxation than~(\ref{eq:basiclp}) as it
has only one variable $\lambda$ for all constraints with the same scope (seen as
a set) of variables. In~(\ref{eq:basiclp}), different constraints have different
variables $\lambda$ even if the scopes (seen as sets) are the same.
Consequently, OSAC solves all problems solved by BLP.
Moreover, since the basic SDP relaxation of a VCSP$(A)$
instance is tighter than BLP~\cite{Raghavendra08:stoc}, 
the basic SDP relaxation also solves all tractable cases identified in this paper.

\end{document}